\documentclass[oneside,reqno]{amsart}
\usepackage[dvips,letterpaper,margin=1.2in]{geometry}

\usepackage[T1]{fontenc}
\usepackage{multirow}
\usepackage{setspace}
\usepackage{amsthm}
\usepackage{bbm}
\usepackage{amsmath}
\usepackage{amssymb}
\usepackage{xcolor}
\usepackage{graphicx}
\usepackage{caption}
\usepackage[colorlinks]{hyperref}
\usepackage{textcomp}
\usepackage{fixmath}
\usepackage{hyperref}
\usepackage[stable]{footmisc}
\usepackage{bm}
\usepackage{color,soul}
\usepackage{subfig}
\usepackage{mathtools}
\usepackage{amssymb}

\usepackage{kantlipsum}
\allowdisplaybreaks
\usepackage{breqn}
\theoremstyle{plain}

\newtheorem{theorem}{Theorem}
\newtheorem{lemma}{Lemma}
\newtheorem{remark}{Remark}
\newtheorem{corollary}{Corollary}
\newtheorem{fact}{Fact}
\newtheorem{assumption}{Assumption}
\newtheorem{definition}{Definition}

\newcommand{\disc}{\mathbf{D}}
\newcommand{\E}{\mathbf{E}}
\newcommand{\N}{N}

\newcommand{\curly}[1]{\left \{ #1 \right\}}
\newcommand{\round}[1]{\left ( #1 \right)}

\newcommand{\RNum}[1]{\uppercase\expandafter{\romannumeral #1\relax}}

\title{THE DISSIPATIVE SPECTRAL FORM FACTOR FOR I.I.D. MATRICES}
\author{Giorgio Cipolloni}
\address{Princeton Center for Theoretical Science, Princeton University, Princeton, NJ 08544, USA}
\curraddr{}
\email{gc4233@princeton.edu}
\thanks{}

\author{Nicolo Grometto}
\address{Sherrerd Hall, Princeton University, Princeton, NJ 08540, USA}
\curraddr{}
\email{ng1069@princeton.edu}
\thanks{}

\subjclass[2010]{Primary }
\date{\today}

\begin{document}
\maketitle

\vspace{-2em}
\begin{abstract}
The \emph{Dissipative Spectral Form Factor} (DSFF), recently introduced in \cite{li2021spectral} for the Ginibre ensemble, is a key tool to study universal properties of dissipative quantum systems. In this work we compute the DSFF for a large class of random matrices with real or complex entries up to an intermediate time scale, confirming the predictions from \cite{li2021spectral}. The analytic formula for the DSFF in the real case was previously unknown. Furthermore, we show that for short times the connected component of the DSFF exhibits a non-universal correction depending on the fourth cumulant of the entries. These results are based on the central limit theorem for linear eigenvalue statistics of non-Hermitian random matrices \cite{cipolloni2021fluctuation, cipolloni2019central}.
\end{abstract}










\section{Introduction}
Non-Hermitian physics has significantly advanced in recent years, leading to a deeper understanding of open (dissipative) quantum systems \cite{deng2010exciton, muller2012engineered, ritsch2013cold, sieberer2016keldysh, chou2011non}, optics \cite{feng2017non, el2018non}, biological systems \cite{may1972will, marchetti2013hydrodynamics}, acoustics \cite{ma2016acoustic, cummer2016controlling}, and many more. The relaxation of the Hermiticity assumption led to the discovery of new interesting phenomena including: non--Hermitian skin--effect \cite{song2019non}, new universality classes \cite{hamazaki2020universality}, dynamical phase transition \cite{kawabata2023entanglement}, replica symmetry breaking in the Sachdev-Ye-Kitaev (SYK) model \cite{garcia2022dominance}, and violation of the Eigenstate Thermalization Hypothesis \cite{cipolloni2023entanglement, cipolloni2023non}. In analogy with the Hermitian case, it is expected that spectral statistics of non-Hermitian systems exhibit universal behavior. More precisely, in \cite{grobe1988quantum, grobe1989universality}, the authors formulated the dissipative analogs of the Berry-Tabor and Bohigas-Giannoni-Schmit conjectures: chaotic systems follow Random Matrix statistics, while integrable systems follow Poisson statistics. To better understand this phenomena, Li, Prosen, and Chan introduced the so called \emph{Dissipative Spectral Form Factor} \cite{li2021spectral} (see also \cite[Section 6.2]{byun2022progress} for a recent review). For a non-Hermitian operator $X$, with complex eigenvalues $\curly{\sigma_i = x_i + iy_i}_{i = 1}^{\N}$, the \emph{Dissipative Spectral Form Factor} (DSFF), introduced in \cite{li2021spectral}, for a complex time parameter $\tau := t + is$, with $t, s \in \mathbf{R}$, is given by
\begin{align}
\label{eq:DSFF}
    \mathbf{DSFF}\round{t,s} := \:\frac{1}{\N^2}\sum_{i,j = 1}^{\N} e^{it \round{x_i - x_j}+is\round{y_i - y_j}}.
\end{align}
In the case when $X$ is a random matrix we consider its expectation
\begin{equation}
\label{eq:defkts}
    K_{\mathbf{F}}(t,s) := \E[\mathbf{DSFF}\round{t,s}],
\end{equation}
with $\mathbf{F}\in \{\mathbf{R},\mathbf{C}\}$ denoting the fact that $X$ has real or complex entries. The DSFF consists of the two dimensional Fourier transform of the two point correlation function of $X$ given by $\rho(z)\rho(z+w)$; in particular, as $\tau$ varies, it studies the correlations of the eigenvalues of $X$ on all scales at the same time. Note that the DSFF reduces to the Hermitian \emph{Spectral Form Factor} (SFF) \cite{leviandier1986fourier}
\begin{equation}
\label{eq:HermSFF}
\mathbf{SFF}(t):=\frac{1}{\N^2}\sum_{i,j = 1}^{\N} e^{it\round{x_i - x_j}},
\end{equation}
when the spectrum of $X$ is real. Additionally, by rewriting $\tau = |\tau|\round{\cos \theta + i \sin \theta}$ and denoting $z_{ij} := \round{x_i-x_j} + i\round{y_i-y_j}$, we may also write $K_{\mathbf{F}}(t,s)=K_{\mathbf{F}}(\tau, \tau^*)=\E \N^{-2} \sum_{i,j = 1}^{\N} e^{i \langle z_{ij}, \tau \rangle}$, which for a fixed angle $\theta$ offers a natural interpretation of the DSFF as SFF of the projection of $\curly{z_{ij}}_{ij}$ onto the radial axis relative to $\theta$. In particular, heuristically, one can think of the DSFF at time $\tau$ as a statistic of the eigenvalues of $X$ which studies the spectrum projected onto the axis relative to $\theta$ on a scale $\sim 1/|\tau|$. Throughout this paper, we will make use of the notation $K_{\mathbf{F}}(\tau, \tau^*)$, rather than $K_{\mathbf{F}}(s,t)$, to stress the dependence on $\tau$ as the underlying complex time parameter.




Before describing some properties of $K(\tau,\tau^*)$, we recall some properties and results about the well known Hermitian SFF \eqref{eq:HermSFF}. As a function of $t$, for chaotic systems, $K(t):=\E[\mathbf{SFF}(t)]$ exhibits the so called \emph{slope--dip--ramp--plateau} behavior (see e.g. \cite[Figure 1]{cipolloni2023spectral}): for short times $K(t)$ decreases with an oscillatory behavior until a "dip--time", $t_{\mathrm{dip}}\sim \N^{1/2}$, then in the regime $\N^{1/2}\lesssim t\lesssim \N$ the SFF grows linearly until the Heisenberg time $t_{\mathrm{Hei}}\sim \N$ when $K(t)$ becomes flat and stays equal to $1/\N$ for $t\ge t_{\mathrm{Hei}}$. We point out that $t_{\mathrm{Hei}}$ is proportional to the inverse level spacing, which is $\sim 1/N$ in the Hermitian case. Despite its great relevance within the physics literature on disordered quantum systems \cite{cotler2017black, garcia2018universality, garcia2017analytical, jia2020spectral, saad2018semiclassical}, the SFF was not mathematically rigorously investigated until very recently. In \cite{forrester2021differential, forrester2021quantifying}, Forrester computed the large $\N$ limit of $K(t)$ for the Gaussian Unitary Ensemble (GUE) and for the Laguerre Unitary Ensemble (LUE), respectively, in the entire \emph{slope--dip--ramp--plateau} regime, relying on the integrable structure of these models (see the remarkable identities in \cite{brezin1997spectral, okuyama2019spectral}). More recently, $K(t)$ has been computed also for the Dyson Brownian motion on the unitary group $U(\N)$ \cite{forrester2022dip}. Only very recently, in \cite{cipolloni2023spectral}, the SFF has been studied for more general Hermitian random matrix models (i.e. for models with entries which are not necessarily Gaussian). In this case, unlike previous works, no exact identities are available, so the SFF was analyzed relying on the recent \emph{multi-resolvent local laws} (see e.g. \cite{cipolloni2022optimal, cipolloni2022thermalisation}), which allowed a rigorous computation of $K(t)$ up to the intermediate time scales $t\ll \N^{5/11}$. The understanding of $K(t)$, for the whole $t\lesssim N$ regime, for matrices with non Gaussian entries is still completely missing.


Much less is known in the harder dissipative (non--Hermitian) case. In \cite{li2021spectral} the authors computed $K_{\mathbf{C}}(\tau,\tau^*)$ for $X$ being a complex Ginibre matrix and they numerically conjectured that a similar behavior is expected for more general classes of non-Hermitian chaotic operators. They also numerically compute the DSFF for $X$ drawn from the real or quaternionic Ginibre ensemble, but no analytical results are available in \cite{li2021spectral} for these cases. In more recent works, the DSFF has been numerically computed also for the dissipative version of the celebrated Sachdev-Ye-Kitaev (SYK) model \cite{garcia2023universality} and for other interacting non-Hermitian systems \cite{ghosh2022spectral}. We also point out that very recently in \cite[Section 5.4]{matsoukas2023non} the authors introduced a different eigenvalue statistic to detect if a non--Hermitian Hamiltonian is chaotic, the \emph{Deformed Spectral Form Factor}\footnote{Not to be confused with the \emph{Dissipative Spectral Form Factor} (abbreviated to DSFF), which is considered in this paper and was introduced in \cite{li2021spectral}.} (see also \cite{matsoukas2023quantum} for its extension to non--Markovian channels).

For concreteness, we now focus only on the complex Ginibre ensemble. According to the predictions in \cite{li2021spectral}, the qualitative behaviour of the DSFF, as a function of $|\tau|$, also follows a \emph{slope--dip--ramp--plateau} behavior, but with fundamental different properties compared to the Hermitian SFF. At leading order, in the large $\text{N}$ limit, the DSFF for the complex Ginibre ensemble is given by
\begin{align}\label{eq:prediction_DSFF_Li}
    K_{\text{GinUE}}(\tau, \tau^*) \approx \frac{1}{\N} + 4\frac{J_1(|\tau|^2)}{|\tau|^2}-\frac{1}{\N}\exp\round{-\frac{|\tau|^2}{4\N}}
\end{align}
where $J_1$ is the Bessel function of the first kind (see Definition~\ref{def:bessel_int} below) and the three terms appearing on the right-hand side of \eqref{eq:prediction_DSFF_Li} are referred to as \emph{contact}, \emph{disconnected}, and \emph{connected} components, respectively. From (\ref{eq:prediction_DSFF_Li}), one notices that the DSFF is rotationally symmetric in the complex time $\tau$, as it depends solely on $|\tau|$.

\begin{figure}[ht]%
    \vspace{-4em}
    \subfloat{{\includegraphics[width=8.5cm]{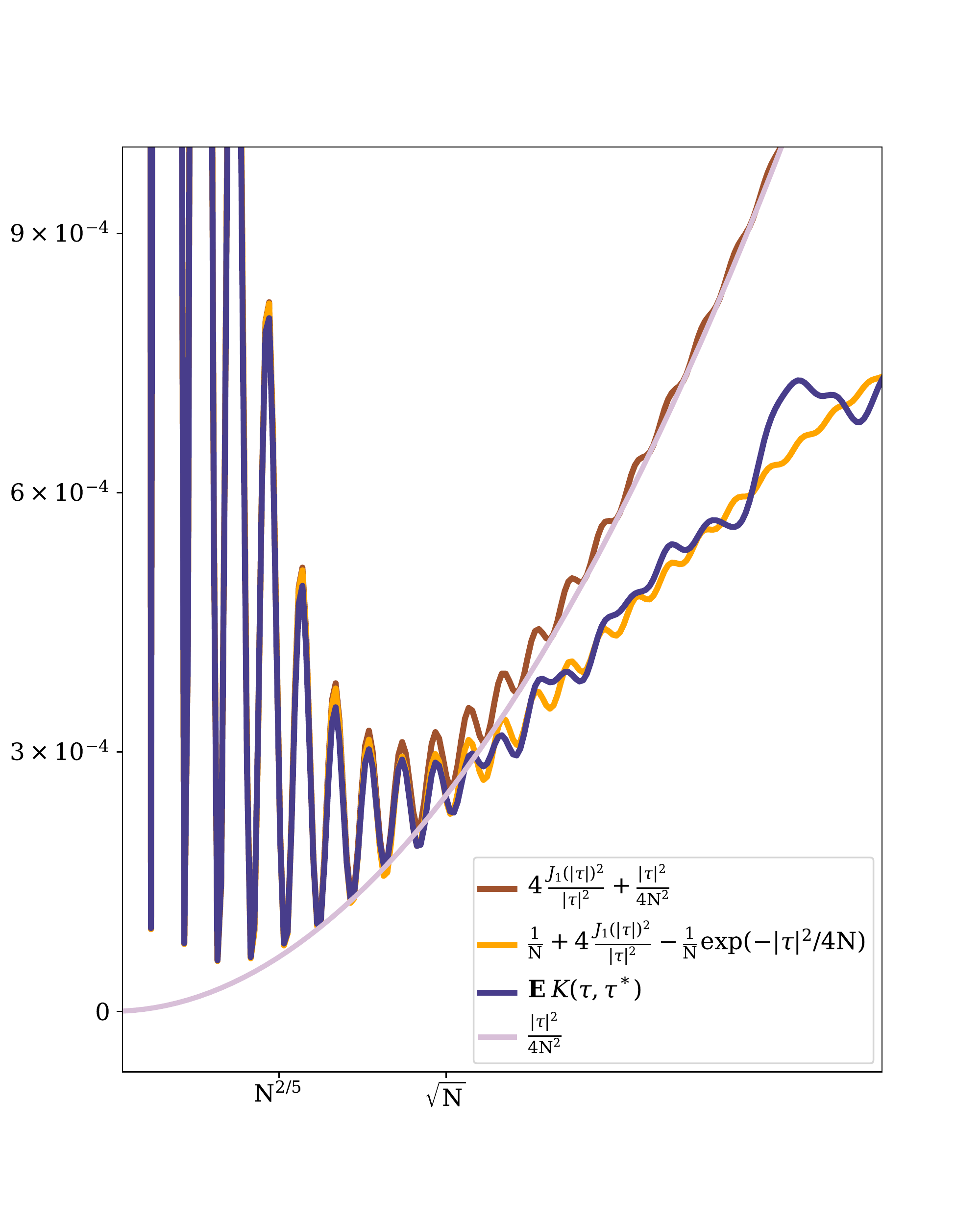} }}%
    \subfloat{{\includegraphics[width=8.5cm]{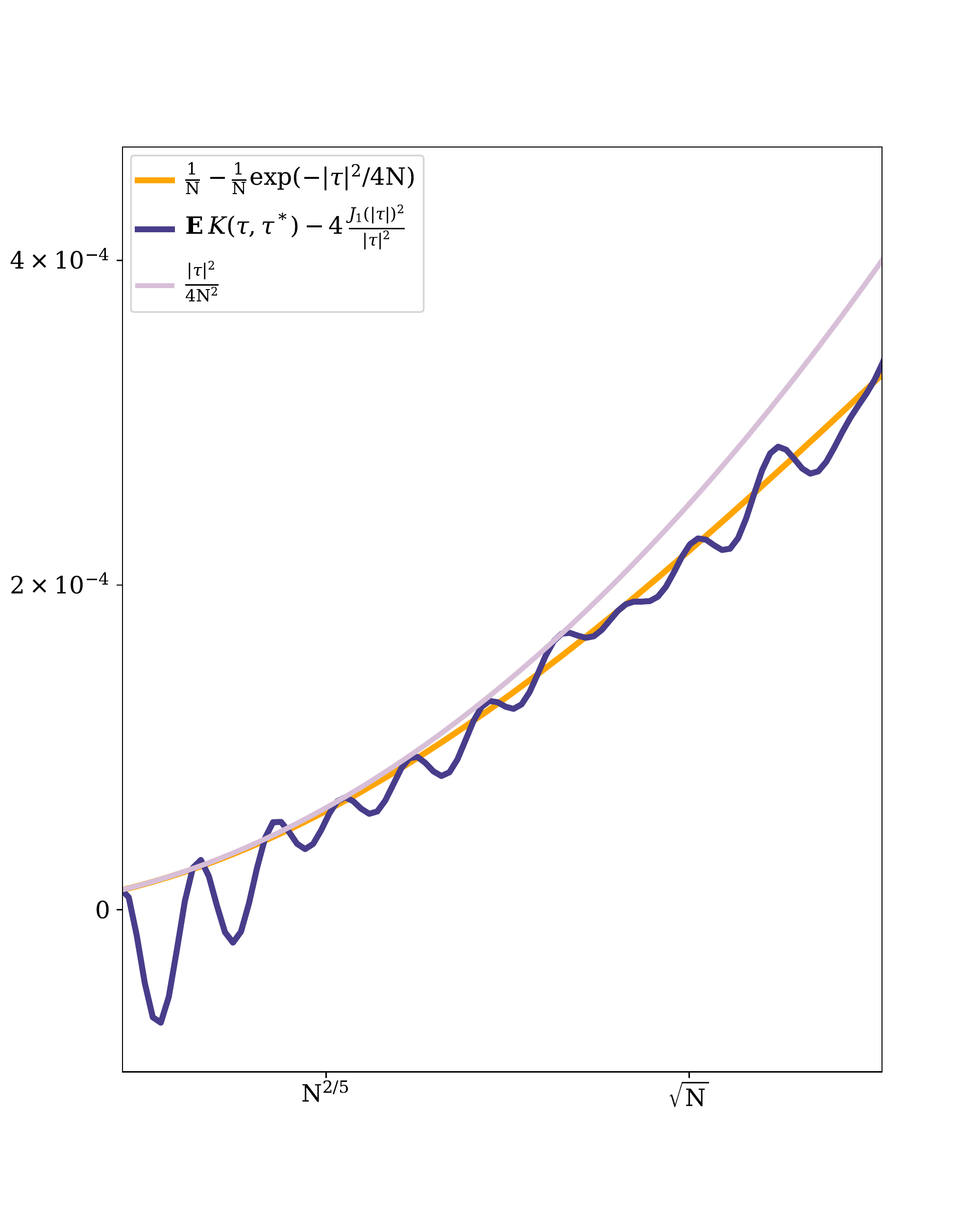} }}%
    \vspace{-3em}
    \caption{\small Numerical results in the complex case for $K(\tau, \tau^*)$ vs $|\tau|$ (matrix size $\N = 1000$, $\theta = 0$, sample size = 1000), displaying a typical realization of the dip-ramp-plateau profile of DSSF. On the right, the curves of interest are plotted without the disconnected component to highlight the quadratic growth of the DSFF in the ramp regime.}
    \label{fig:DSFF_numerics}
\end{figure}


Figure~\ref{fig:DSFF_numerics} below shows the \emph{slope--dip--ramp--plateau} behavior for the DSFF. By standard asymptotic of Bessel functions (see e.g. Fact~\ref{fact:bessel_properties} below) we notice that the Heisenberg time scales as the inverse of the mean eigenvalue spacing, that is $\tau_{\text{Hei}}\sim  \sqrt{\N}$ (in analogy with the Hermitian SFF when $t_{\mathrm{Hei}}\sim\N$). In addition, by relying on the fact that the non-oscillatory part of the disconnected component asymptotically scales as $ |\tau|^{-3}$, for $|\tau|\gg 1$, and by considering the time at which the disconnected and connected contributions are of the same order, one obtains that $\tau_{\text{edge}} \sim \N^{2/5}$ (this is the analog of $t_{\mathrm{dip}}$ in the Hermitian case). Furthermore, using that that $J_1(z)/z  \sim 1$ as $z \to 0$, one also deduces that the initial decay of the DSFF from $K(0,0) = 1$ for $|\tau|\lesssim \tau_{\text{edge}}$ is governed by the disconnected component. In the intermediate regime, for $\tau_{\text{edge}} \lesssim |\tau|\: \lesssim \tau_{\text{Hei}}$, the DSFF increases quadratically at a rate $|\tau|^2/4\N^2$, which may be seen by combining the contact component with the Taylor expansion of the exponential term around zero. This is in stark contrast to the Hermitian SFF, where the intermediate ramp behaviour exhibits linear growth with respect to time.  Finally, at time $|\tau|\:\gtrsim \tau_{\text{Hei}}$, the DSFF  has a  plateau at the mean eigenvalue spacing $1/\sqrt{\N}$, in analogy to the Hermitian case.

In \cite{li2021spectral} the authors analitically computed the asymptotic of the DSFF only in the complex Ginibre ensemble relying on its integrable structure \cite{ginibre1965statistical, mehta2004random}; for the real Ginibre ensemble only numerics are available. Although the joint distribution of the eigenvalues is explicitly known also for the real Ginibre ensemble \cite{borodin2009ginibre, forrester2007eigenvalue} (see also the recent review \cite{byun2023progress}), it is much more involved due to the special role played by the real axis, thus making its analysis less amenable. The aim of this work is to consider non-Hermitian matrices with generic (i.e. not necessarily Gaussian) real or complex entry distribution (see Assumption~\ref{asspt:iid_matrix} below). In particular, in Theorem~\ref{thrm:dsff_leading_gen} below, we rigorously prove \eqref{eq:prediction_DSFF_Li} for a large class of non-Hermitian matrices with complex entries up to intermediate time scales $|\tau|\ll \N^{2/7}$. Additionally, we give an analog of \eqref{eq:prediction_DSFF_Li} for matrices with real entries which we conjecture to hold up to $|\tau|\ll \sqrt{\N}$ (and prove for $|\tau|\ll \N^{2/7}$). More precisely, we conjecture that for $1\ll |\tau|\ll \sqrt{\N}$ it holds
\begin{equation}
\label{eq:newexpr}
    K_{\mathbf{F}}(\tau,\tau^*)\approx 4 \frac{J_1(|\tau|)^2}{|\tau|^2} + \frac{|\tau|^2}{4\N^2}+\frac{(t^2-s^2)(2/\beta-1)}{4s \N^2}J_1(2s),
\end{equation}
with $\beta$ a parameter such that $\beta=1$ in the real case and $\beta=2$ in the complex case; furthermore, $\mathbf{F}=\mathbf{R}$ for $\beta=1$ and $\mathbf{F}=\mathbf{C}$ for $\beta=2$. In particular, we show that also in the non-Hermitian case the DSFF can be used to distinguish different universality classes. We remark that the expression \eqref{eq:newexpr} in the real case was unknown even for the real Ginibre ensemble. Furthermore, in Theorem~\ref{thrm:dsff_leading_gen} below, we show that for matrices $X$ with general entry distribution there is an additional correction to the connected component of $K_{\mathbf{F}}(\tau,\tau^*)$ for $|\tau|\sim 1$ depending on the fourth moment of the entries. We remark that to compute this asymptotic for the DSFF we rely on the recent CLT-type results for linear eigenvalue statistics of non-Hermitian random matrices appearing in \cite{cipolloni2019central, cipolloni2021fluctuation} (see also \cite{cipolloni2022mesoscopic, forrester1999fluctuation, forrester2009matrix, kopel2015linear, nourdin2010universal, o2016central, rider2006gaussian, rider2007noise, tao2015random} for other CLT results for linear statistics of non--Hermitian eigenvalues).

\subsection{Notation and Conventions} For positive quantities $f,g$, the notation $f\sim g$ is used to indicate asymptotic equivalence up to multiplicative constants, i.e. that there exist constants $c,C > 0$, such that $c \leq f/g \leq C$. When $c,C$ may be taken to be both equal to one in some specified limit, we write $f\approx g$. Analogously, $f \ll g, \: f \lesssim g$ are used to indicate that $f/g \to 0, \: f/g \leq C$, respectively. In addition, we make use of standard asymptotic notation, according to $f = O(g), \: f = o(g)$. By $\mathbf{D}\subset \mathbf{C}$ we denote the open unit disk, and for any $z \in \mathbf{C}$, we use the notation $d^2 z := 2^{-1}i\round{dz  \wedge  \overline{z}}$ to denote the two dimensional volume form on $\mathbf{C}$.

\section{Main Results}

We consider $\N\times \N$ non-Hermitian matrices $X$ satisfying the following assumption.

\begin{assumption}\label{asspt:iid_matrix}
Let $X$ be an $\N \times \N$ matrix with real or complex independent identically distributed (i.i.d.) entries $X_{ij} \stackrel{d}{=} \N^{-1/2}\chi$. The random variable $\chi$ is such that $\E \chi =0$, $\E |\chi|^2=1$; additionally, in the complex case we also assume that $\E \chi^2 = 0$. Furthermore, we assume that for all $p \in \mathbf{N}$, there exists a constant $C_p > 0$ such that $\E |\chi|^p \leq C_p$.
\end{assumption}






The main result of this paper is to prove the asymptotic of the DSFF for a large class of models satisfying Assumption~\ref{asspt:iid_matrix}, up to some intermediate time $|\tau|\ll \N^{2/7}$.






\begin{theorem}\label{thrm:dsff_leading_gen}
Let $X$ be a real or complex i.i.d. matrix satisfying Assumption~\ref{asspt:iid_matrix}. For $s,t \in \mathbf{R}$, and $\tau = t+is$ let $K_{\mathbf{F}}(\tau,\tau^*)=K_{\mathbf{F}}(t,s)$ be defined as in \eqref{eq:DSFF}. Then for $0 \le |\tau|\ll \N^{2/7}$ we have
\begin{equation}
\label{eq:DSFFmaintheonew}
    K_{\mathbf{F}}(\tau,\tau^*)=\left[e(\tau,\tau^*)^2+\frac{v(\tau,\tau^*)}{\N^2}\right](1+o(1)), \qquad\quad \mathbf{F}\in \{\mathbf{R},\mathbf{C}\},
\end{equation}
where
\begin{equation}
\begin{split}
\label{eq:expvar}
  e(\tau,\tau^*):&=  2\frac{ J_1(|\tau|)}{|\tau|}- \frac{ |\tau|J_1(|\tau|)}{4N}+ 4\kappa_4 \frac{J_3\round{|\tau|}}{\N |\tau|} \\
  &\quad+\frac{2/\beta-1}{\N}\left(\frac{1}{4\pi}\normalcolor\int_{\mathbf{D}}\frac{e^{itx}(1-e^{isy})}{y^2}dxdy-J_0(|\tau|)+\frac{J_0(t)}{2}+\frac{\cos(t)}{2}\right) \\
  v(\tau,\tau^*):&= \frac{|\tau|^2}{4}+ \kappa_4 \round{2\frac{J_1(|\tau|)}{|\tau|} - J_0\round{|\tau|}}^2+ \sum_{k \in \mathbf{Z}}|k|[\beta-1+(2\beta-1)|\sin(\varphi k)|^2] J_k\round{|\tau|}^2 \\
  &\quad+\frac{(t^2-s^2)(2/\beta-1)}{4s}J_1(2s),
\end{split}
\end{equation}
with $\kappa_4 := \E |\chi|^4 -(1+2/\beta)$ denoting the fourth cumulant of the entries of $X$,  and the angle $\varphi=\varphi(t,s)$ defined so that
\[
\sin \varphi = t/\sqrt{t^2+s^2} \hspace{3em} \cos \varphi = s/\sqrt{t^2+s^2}.
\]
Here $\beta$ is a parameter such that $\beta=1$ in the real case and $\beta=2$ in the complex case.
\end{theorem}

We remark that for $|\tau|\lesssim 1$ the connected component of the DSFF in \eqref{eq:DSFFmaintheonew} depends on $\kappa_4$, i.e. it is sensitive to the fourth moment of the distribution of the entries of $X$. This shows that the DSFF for fairly short times deviates from the Ginibre ensemble (when $\kappa_4=0$) given in \cite[Eq. (4)]{li2021spectral}.

Next, we show that the expression in \eqref{eq:DSFFmaintheonew} substantially simplifies for $|\tau|\gg 1$. In particular, in this regime there is no dependence on the fourth cumulant $\kappa_4$, but there is still a substancial difference between the complex and the real case.

\begin{corollary}\label{thrm:dsff_leading}
Let $X$ be a real or complex i.i.d. matrix satisfying Assumption~\ref{asspt:iid_matrix}. For $s,t \in \mathbf{R}$, and $\tau = t+is$ let $K_{\mathbf{F}}(\tau,\tau^*)=K_{\mathbf{F}}(t,s)$ be defined as in \eqref{eq:DSFF}. Then for $1\ll|\tau|\ll \N^{2/7}$ we have
\begin{equation}
\label{eq:DSFFmaintheo}
    K_{\mathbf{F}}(\tau,\tau^*)=\left[4 \frac{J_1(|\tau|)^2}{|\tau|^2} +\frac{1}{\N^2}\left( \frac{|\tau|^2}{4}+\frac{(t^2-s^2)(2/\beta-1)}{4s}J_1(2s)\right)\right] (1+o(1)), \qquad \mathbf{F}\in \{\mathbf{R},\mathbf{C}\},
\end{equation}
where $\beta$ is a parameter such that $\beta=1$ in the real case and $\beta=2$ in the complex case.
\end{corollary}

Note that in the regime $1\ll |\tau|\ll \sqrt{\N}$ the DSFF $K_{\mathbf{F}}(\tau,\tau^*)$ consists of the sum of two different terms. The first term (which dominates for $|\tau|\ll N^{2/5}$) comes from the the global density of the eigenvalues, which is model dependent, whilst the second term (which dominates for $N^{2/5}\ll |\tau|\ll \sqrt{N}$) is expected to be universal, i.e. depends only on the symmetry class of the matrix $X$. Furthermore, note that $K_{\mathbf{C}}(\tau,\tau^*)$ is rotationally symmetric, whilst $K_{\mathbf{R}}(\tau,\tau^*)$ is not, as a consequence of the symmetry of the spectrum with respect to the real axis. In particular, we note that in the time regime $1\ll|\tau|\ll\sqrt{\N}$, the expression for the DSFF obtained in \eqref{eq:DSFFmaintheo} agrees with \cite[Eq. (4)]{li2021spectral} in the complex case; this follows by considering the Taylor expansion of $x\mapsto e^{-x}$ around $x = 0$, which yields 
\begin{align*}
    K_{\text{GinUE}}(\tau, \tau^*) =\N+4\N^2\frac{J_1(|\tau|)^2}{|\tau|^2}-\N\exp\round{-\frac{|\tau|^2}{4\N}} = 4\N^2\frac{J_1(|\tau|)^2}{|\tau|^2} + \frac{|\tau|^2}{4}+o\round{\frac{|\tau|^2}{4}}. 
\end{align*}
Additionally, in the real case, for $N^{2/5}\ll |\tau|\ll \sqrt{N}$ (i.e. we only consider the connected component) we have $K_{\mathbf{R}}(\tau,\tau^*)=2K_{\mathbf{C}}(\tau,\tau^*)$ for $\theta=0$, which follows from \eqref{eq:DSFFmaintheo} by $J_1(2s)/s\to 1$ as $s\to 0$, and $K_{\mathbf{R}}(\tau,\tau^*)=K_{\mathbf{C}}(\tau,\tau^*)$ for $\theta\in (0,\pi/2]$, confirming numerical predictions from \cite[Appendix A]{li2021spectral}. We point out that in \cite[Appendix A]{li2021spectral} the authors notice a different behavior of $K_{\mathbf{R}}(\tau,\tau^*)$ when $\theta=\pi/2$ as well, however this phenomenon is not visible on the time scales $|\tau|\ll \sqrt{N}$ we consider here, since this different behaviour is caused by the degeneracy of the spectrum due to the $\sim\sqrt{N}$ real eigenvalues, which would be visible only at scales $|\tau|\gtrsim \sqrt{N}$. On the other hand, we are able to detect the transition of $K_{\mathbf{R}}(\tau,\tau^*)$ as $\theta\to 0$ since this effect is caused by the $2$--fold degeneracy of $\sim N$ eigenvalues (i.e. each complex eigenvalue is counted twice).

\begin{remark}
In Theorem~\ref{thrm:dsff_leading_gen} (and in Corollary~\ref{thrm:dsff_leading}) we computed only the expectation of the DSFF; however, our method, relying on \cite{cipolloni2021fluctuation, cipolloni2019central}, also allows to compute higher moments
\[
\E \left|\frac{1}{\N^2}\sum_{i,j = 1}^{\N} e^{it \round{x_i - x_j}+is\round{y_i - y_j}}\right|^k \qquad \mathrm{for} \qquad k\in\mathbf{N}.
\]
We refrain from doing this here as it is out of the scope of the current paper.
\end{remark}

We conclude this section pointing out that the methods used in the current paper allow us to rigorously prove \eqref{eq:DSFFmaintheo} only up to the intermediate time $|\tau|\ll \N^{2/7}$; however, we expect \eqref{eq:DSFFmaintheo} and our proof method in Section~\ref{sec:CLT}, to be correct up to the Heisenberg time, i.e. up to $|\tau|\ll \tau_{\mathrm{Hei}}\sim\sqrt{\N}$ (this is also confirmed numerically in Figure~\ref{fig:DSFF_numerics}).


\vspace{1em}

\section{Asymptotic of the DSFF: Proof of Theorem~\ref{thrm:dsff_leading}}
\label{sec:CLT}

 
We start by rewriting the DSFF as a linear statistic of the eigenvalues of $X$ for a specific test function. For this purpose we introduce the function
\begin{equation}
    \label{eq:defftau}
    f_{\tau,\tau^*}(z) := e^{it \Re z + is \Im z} = e^{it\frac{z+\overline{z}}{2}+s\frac{z-\overline{z}}{2}} = e^{i\tau\frac{z^*}{2}+i\tau^*\frac{z}{2}}
\end{equation}
for $\tau = s+it$, and $s,t \in \mathbf{R}$. We can thus write the averaged DSFF as follows
\begin{align}\label{eq:Expect_Kst}
     \N^2 \cdot K_{\mathbf{F}}(\tau, \tau^*) = \E\left[ \left | \sum_{i = 1}^{\N} f_{\tau, \tau^*}(\sigma_i) \right |^2\right] = \left |\E \left[\sum_{i = 1}^{\N} f_{\tau, \tau^*}(\sigma_i)\right]\right |^2 + \mathbf{Var}\left[\sum_{i = 1}^{\N} f_{\tau, \tau^*}(\sigma_i)\right], 
\end{align}
with $\E$ and $\mathbf{Var}$ denoting the expectation and the variance with respect to the random matrix $X$. From now on, for concreteness, we focus on the proof in the complex case. The computations in the real case are similar and rely on \cite[Theorem 2.2]{cipolloni2021fluctuation} rather than \cite[Theorem 2.2]{cipolloni2019central}; we present the main difference in Section~\ref{sec:realcase} below.

Define
\[
L_{\N}(f):=\sum_{i = 1}^{\N} f(\sigma_i) - \E \sum_{i = 1}^{\N} f(\sigma_i)\Rightarrow L(f).
\]
Then, by \cite[Theorem 2.2]{cipolloni2019central} (together with \cite[Corolalry 2.4]{cipolloni2019central} for the effective convergence of moments\footnote{We remark that in \cite[Corolalry 2.4]{cipolloni2019central} the dependence on $\lVert \Delta f \rVert_2$ is not explicitly written, but it can be deduced by inspection of the proof.}), for any sufficiently smooth test function $f$ we have\footnote{Here $\lVert \cdot \rVert_2$ denotes the usual $L^2$--norm. Furthermore, for $h$ defined on the boundary of the unit disk $\partial \mathbf{D}$, we define its Fourier transform by
\[
    \widehat{h}(k):= \frac{1}{2\pi}\int_0^{2\pi} h(e^{i\theta}) e^{-i\theta k}d\theta, \qquad\quad k \in \mathbf{Z}.
\]} (recall that $\kappa_4$ denotes the fourth cumulant of the entries)
\begin{equation}
\begin{split}
\label{eq:usefulrel}
\E L_{\N}(f)&=\frac{\N}{\pi}\int_{\mathbf{D}} f(z) d^2z +\frac{1}{8\pi} \int_{\mathbf{D} }\Delta f(z)\, d^2z  - \frac{\kappa_4}{\pi} \int_{\mathbf{D}} f(z)\round{2|z|^2-1}d^2 z  +O\left(\frac{\lVert \Delta  f\rVert_2}{\N^c}\right)\\
\E|L_{\N}(f)|^2&=\frac{1}{4\pi}\int_{\mathbf{D}} \left|\nabla  f\right|^2\, d^2z+\frac{1}{2}\sum_{k\in \mathbf{Z}} |k|\left|\widehat{ f|_{\partial\mathbf{D}}}(k)\right|^2\\
& \quad +\kappa_4 \left|\frac{1}{\pi}\int_{\mathbf{D}}f(z)d^2 z - \frac{1}{2\pi} \int_0^{2\pi} f(e^{i\theta}) d\theta\right|^2+O\left(\frac{\lVert \Delta f\rVert_2^2}{\N^c}\right),
\end{split}
\end{equation}
for some small fixed $c>0$.

For $\tau$ as prescribed above, $f_{\tau, \tau^*}$ defined in \eqref{eq:defftau} satisfies the assumptions of the above CLT, so in particular it satisfies \eqref{eq:usefulrel} with $\lVert \Delta f_{\tau,\tau^*}\rVert_2\lesssim |\tau|^2$. We now compute the explicit terms in the rhs. of \eqref{eq:usefulrel} when $f=f_{\tau,\tau^*}$. For this purpose, we recall the definition of Bessel functions of the first kind.
\begin{definition}\label{def:bessel_int}
For $n\in\mathbf{Z}, z\in \mathbf{C}$, we define the $n$-th Bessel function of the first kind by
\[J_n(z):= \frac{1}{2\pi} \int_{-\pi}^{\pi} e^{i\round{n\tau - z\sin \tau }}d\tau.  \]
\end{definition}

Next, we will use several important properties of Bessel functions (see \cite[Section 9]{abramowitz1988handbook}), which we gather below for the reader's convenience.
\begin{fact}
\label{fact:bessel_properties} For $k,l \in \mathbf{N}$, $n\in \mathbf{Z}$, and $z, \omega, u \in \mathbf{C}$, we have
\begin{enumerate}
    \item[(i)] $J_n$ may equivalently be defined by the following series 
\[J_n(z):=\sum_{m = 0}^{\infty}\frac{(-1)^m}{m!\Gamma(m+n+1)}\round{\frac{z}{2}}^{2m + n};\]
\item[(ii)] $J_{-n}(z) = (-1)^n J_n(z)$;
\item[(iii)] For $|z| \gg |k^2-\frac14|$, $J_k(z) = \sqrt{\frac{2}{\pi z}}\cos\round{z-\frac{k\pi}{2}-\frac{\pi}{4}}$(1+o(1));
\item[(iv)] $\sum_{k \in \mathbf{Z}} J_k(\omega)J_k(u) e^{-ik\theta} = J_0\round{\sqrt{\omega^2+u^2-2\omega u \cos \theta}}$, \quad for\,\, $\theta \in (-\pi, \pi]$;
\item[(v)] $\round{\frac1z \frac{d}{dz}}^l \round{z^k J_k(z)} = z^{k-l}J_{k-l}(z)$.
\end{enumerate}
\end{fact}

By relying on the above, we proceed to derive expressions for the desired quantities in (\ref{eq:Expect_Kst}); the proof of this lemma is postopned at the end of this section.
\begin{lemma}\label{lemma:mean_var} For $s,t \in \mathbf{R}$, and $\tau = t+is$, there exists $c>0$ such that
\begin{align}\label{eq:mean_var_expressions}
    \begin{split}
    \E \left[\sum_{i = 1}^{\N} f_{\tau,\tau^*}(\sigma_i) \right]& = 2\N\frac{ J_1(|\tau|)}{|\tau|}- \frac{ |\tau|J_1(|\tau|)}{4}+ 4\kappa_4 \frac{J_3\round{|\tau|}}{|\tau|} + O\left(\frac{|\tau|^2}{\N^c}\right)\\
    \mathbf{Var}\left[\sum_{i = 1}^{\N} f_{\tau,\tau^*}(\sigma_i)\right] & = \frac{|\tau|^2}{4}+ \frac12\sum_{k \in \mathbf{Z}}|k| J_k\round{|\tau|}^2 + \kappa_4\round{ 2\frac{J_1(|\tau|)}{|\tau|} - J_0\round{|\tau|}}^2+ O\left(\frac{|\tau|^4}{\N^c}\right).
    \end{split}
\end{align}
\end{lemma}

We are now ready to conclude Theorem~\ref{thrm:dsff_leading_gen}.


\begin{proof}[Proof of Theorem~\ref{thrm:dsff_leading_gen}]
The asymptotic in \eqref{eq:DSFFmaintheonew} immediately follows from Lemma~\ref{lemma:mean_var}.
\end{proof}

Then, in order to conclude the asymptotic in Corollary~\ref{thrm:dsff_leading} from Theorem~\ref{thrm:dsff_leading_gen}, i.e. to identify the leading term of \eqref{eq:DSFFmaintheonew} for $|\tau|\gg 1$, we will use the following additional technical lemma, whose proof is presented at the end of this section.


\begin{lemma}\label{lemma:series_control}
    For $x > 0$ we have $\sum_{k \in \mathbf{Z}}|k| J_k\round{x}^2 \lesssim x$. 
\end{lemma}

\begin{proof}[Proof of Corollary~\ref{thrm:dsff_leading}]
 Upon squaring the result obtained for the expectation in (\ref{eq:mean_var_expressions}) and using (iii) of Fact \ref{fact:bessel_properties}, we readily obtain that $\left |\E \sum_{i} f_{\tau,\tau^*}(\sigma_i)\right |^2 \sim 4 \N^2J_1(|\tau|)^2/|\tau|^2$. In addition, by Lemma~\ref{lemma:series_control} we also have that for $|\tau| \gg 1$, the leading terms in the expression for the variance in (\ref{eq:mean_var_expressions}) is $|\tau|^2/4$. This yields the desired result, upon recalling (\ref{eq:Expect_Kst}).
\end{proof}

We now conclude this section with the proof of Lemmas~\ref{lemma:mean_var}--\ref{lemma:series_control}.

\begin{proof}[Proof of Lemma~\ref{lemma:mean_var}]
We start with the computation of the expectation in the first line of \eqref{eq:usefulrel}. Recall the definition of $f_{\tau,
\tau^*}(z)$ from \eqref{eq:defftau}, then using the parametrizations $z=r(\cos \theta+i\sin \theta)$, we obtain
\begin{equation}
\label{eq:polarcor}
    \int_{\mathbf{D}}f_{\tau,\tau^*}(z)d^2z = \int_{x^2 + y^2 < 1} e^{itx + isy}dx dy  = \int_0^1 r \int_{-\pi}^\pi  e^{i\round{t r \cos \theta + s r\sin \theta}}\, d\theta dr.
\end{equation}

By standard double-angle identities, we note that we may write
\begin{equation}\label{eq:double_angle}
        t \cos \theta + s \sin \theta = \sqrt{t^2 + s^2}\sin\round{\varphi + \theta} 
\end{equation}
for a unique $\varphi \in (-\pi, \pi]$, s.t. $\sin\varphi = t/\sqrt{t^2+s^2},\: \cos\varphi = s/\sqrt{t^2+s^2}$. This implies that the inner integral in \eqref{eq:polarcor} can be written as
\begin{equation}
\label{eq:compint}
\begin{split}
    \int_{-\pi}^\pi e^{i\round{tr \cos \theta + sr \sin \theta}} d\theta & = \int_{-\pi}^\pi e^{ir \sqrt{t^2 + s^2}\sin\round{\varphi + \theta}}d\theta  \\
      & = \int_{-\pi}^{\pi} e^{-ir \sqrt{t^2 + s^2}\sin\round{\theta}}d\theta  \\
      & = 2\pi J_0\round{r\sqrt{t^2+s^2}},
      \end{split}
\end{equation}
where in the second equality we used that, by periodicity, the integral over $(-\pi+\varphi,\pi+\varphi]$ is equal to the one over $(-\pi,\pi]$. Plugging \eqref{eq:compint} into \eqref{eq:polarcor}, and using the series expansion from (i) of Fact~\ref{fact:bessel_properties}, we obtain 
\begin{equation}
    \begin{split}\label{eq:first_int_expectation}
     \int_{\mathbf{D}}f_{\tau,\tau^*}(z)d^2z & = 2\pi\int_0^1 r  J_0\round{r\sqrt{t^2+s^2}} dr \\
      &\quad  = 2\pi\sum_{m = 0}^{\infty}\frac{(-1)^m}{m!\Gamma(m+1)}\round{\frac{\sqrt{t^2+s^2}}{2}}^{2m}\int_0^1 r^{2m+1} dr \\
     & \quad = \frac{2\pi}{\sqrt{t^2+s^2}}\sum_{m = 0}^{\infty}\frac{(-1)^m}{m!\Gamma(m+2)}\round{\frac{\sqrt{t^2+s^2}}{2}}^{2m+1}\\
    & \quad =  2\pi\frac{J_1(|\tau|)}{|\tau|}.
    \end{split}
\end{equation}
By a similar argument, we compute the third integral in the first line of \eqref{eq:usefulrel}
\begin{align*}
    \int_{\disc} f_{\tau,\tau^*}(z)\round{2|z|^2-1}d^2z & =  \int_0^1 \int_{ -\pi}^\pi e^{itr\cos \theta + isr\sin \theta}\round{2r^3-r}d\theta dr \\
    & = \int_0^1 \round{2r^3-r} \int_{-\pi}^\pi e^{i\round{tr \cos \theta + sr \sin \theta}} d\theta dr.
\end{align*}
Using again \eqref{eq:polarcor} to compute the $\theta$--integral, we obtain
\begin{align}\label{eq:second_int_expectation}
    \begin{split}
     \int_{\disc} f_{\tau,\tau^*}(z)\round{2|z|^2-1}d^2z & = 2\pi \int_0^1 \round{2r^3-r} J_0\round{r\sqrt{t^2+s^2}}dr \\
     & \quad = -2\pi\round{\frac{\sqrt{t^2+s^2}}{2}}^2 \sum_{m=1}^{\infty}\frac{(-1)^{m-1}}{(m-1)!\Gamma(m+3)}\round{\frac{\sqrt{t^2+s^2}}{2}}^{2(m-1)} \\
     & \quad = -2\pi\round{\frac{\sqrt{t^2+s^2}}{2}}^{-1} \sum_{m=0}^{\infty}\frac{(-1)^m}{m!\Gamma(m+4)}\round{\frac{\sqrt{t^2+s^2}}{2}}^{2m+3} \\
     & \quad = -4\pi\frac{ J_3\round{|\tau|}}{|\tau|}.
    \end{split}
\end{align}

Combining \eqref{eq:first_int_expectation} and \eqref{eq:second_int_expectation}, and using that $\Delta f_{\tau,\tau^*}=-|\tau|^2f_{\tau,\tau^*}$, yield the desired result for the expectation term in \eqref{eq:mean_var_expressions}.


Next, we consider the terms in the variance in \eqref{eq:usefulrel}. We start from the term consisting of the square of the $L^2(\mathbf{D})$--norm of $\nabla f$, which is given by
\begin{equation}
\label{eq:gradient}
   \frac{1}{4\pi}\left\|\nabla f_{\tau,\tau^*}\right\|_{L^2(\mathbf{D})}^2 =  \frac{|\tau|^2}{4}.
\end{equation}
For the second term in \eqref{eq:usefulrel}, choosing $\varphi$ as in \eqref{eq:double_angle}, we have
    \begin{equation}
    \label{eq:comph12}
    \begin{split}
        \hat{f}(k)  = \frac{1}{2\pi}\int_{0}^{2\pi} e^{it\cos \theta + is \sin \theta - i\theta k}d\theta &= \frac{e^{i\varphi k}}{2\pi}\int_{\varphi}^{\varphi + 2\pi}e^{i\round{-k \alpha + \sqrt{t^2+s^2}\sin\alpha}}d\alpha\\
        & = e^{i\varphi k} J_{-k}\round{-\sqrt{t^2+s^2}} \\
        &= e^{i\varphi k}(-1)^{k+1} J_{k}\round{|\tau|},
        \end{split}
    \end{equation}
where in the last step we used (ii) of Fact~\ref{fact:bessel_properties} and the expansion of $J_k$ in (i) of Fact~\ref{fact:bessel_properties}. We thus obtain
    \begin{align}\label{eq:holder_norm}
        \frac{1}{2}\sum_{k\in \mathbf{Z}} |k|\left|\widehat{(f_{\tau,\tau^*})|_{\partial\mathbf{D}}}(k)\right|^2 =\frac{1}{2}  \sum_{k \in \mathbf{Z}}|k| J_k\round{|\tau|}^2.
    \end{align}

Using computations analogous to the ones used to obtain \eqref{eq:first_int_expectation}--\eqref{eq:second_int_expectation}, we get
\begin{equation}
\label{eq:kappa4}
\begin{split}
    \frac{1}{\pi}\int_{\mathbf{D}} f_{\tau, \tau^*}(z) d^2 z & = 2\frac{J_1\round{|\tau|}}{|\tau|}\\
 \frac{1}{2\pi} \int_0^{2\pi} f_{\tau, \tau^*}(e^{i\theta}) d\theta & = \sum_{n\in \mathbf{Z}}(-i)^n J_n(t)J_n(s) = J_0(\sqrt{t^2+s^2})
    \end{split}
\end{equation}
where the last equality follows from (iv) of Fact~\ref{fact:bessel_properties}, upon choosing $\theta = \pi/2$.
Finally, combining \eqref{eq:gradient}--\eqref{eq:kappa4}, we obtain the desired expression for the variance in \eqref{eq:mean_var_expressions}.
\end{proof}

\begin{proof}[Proof of Lemma~\ref{lemma:series_control}]
    By the Cauchy-Schwarz inequality, together with the fact that $\sum_{k \in \mathbf{Z}}J_k(x)^2 = J_0(0) = 1$ (see e.g. (v) of Fact~\ref{fact:bessel_properties} for $\theta= 0$ and $\omega=u$), we obtain
    \begin{align}\label{eq:cs_series}
        \sum_{k\in \mathbf{Z}}|k|J_k(x)^2 \leq \sqrt{\sum_{k\in {\mathbf{Z}}}k^2 J_k(x)^2}.
    \end{align}
    Next, by differentiating both sides of the expression in (v) of Fact~\ref{fact:bessel_properties} with respect to $\theta$, and relying on the differentiation rule for Bessel functions of the first kind in (vi) of Fact~\ref{fact:bessel_properties}, we obtain 
    \begin{align}
    \begin{split}
    \label{eq:derrel}
        \sum_{k \in \mathbf{Z}} k^2 J_k(x)^2 e^{-ik\theta} & = J_0\round{x\sqrt{2\round{1-\cos \theta}}}\frac{x^2 \sin^2\theta}{2 \round{1-\cos \theta} } + \frac{J_1\round{x\sqrt{2\round{1-\cos \theta}}}}{x\sqrt{2\round{1-\cos \theta}}}\round{\cos \theta - \frac{\sin^2\theta}{1-\cos \theta}}x^2.
    \end{split}
    \end{align}
    Hence, in the limit $\theta \rightarrow 0$, using that $J_1(z)/z \to 1/2$ for $z\to 0$, the relation in \eqref{eq:derrel} yields
    \begin{align*}
         \sum_{k \in \mathbf{Z}} k^2 J_k(x)^2 = \frac{x^2}{2},
    \end{align*}
    which together with \eqref{eq:cs_series} gives the desired result. 
\end{proof} 

We conclude this section with the computations in the real case.

\subsection{DSFF in the real case}
\label{sec:realcase}

By \cite[Theorem 2.2]{cipolloni2021fluctuation} we have (recall that $\kappa_4$ denotes the fourth cumulant of the entries)
\begin{equation}
\begin{split}
\label{eq:usefulrelreal}
\E L_{\N}(f)&=\frac{\N}{\pi}\int_{\mathbf{D}} f(z) d^2z - \frac{\kappa_4}{\pi} \int_{\mathbf{D}} f(z)\round{2|z|^2-1}d^2 z +\frac{1}{4\pi}\int_{\mathbf{D}}\frac{f(\Re z)-f(z)}{(\Im z)^2}\,d^2z +\frac{1}{8\pi} \int_{\mathbf{D} }\Delta f(z)\, d^2z\\
&\quad-\frac{1}{2\pi}\int_0^{2\pi} f(e^{i\theta})\,d\theta+\frac{1}{2\pi}\int_{-1}^1\frac{f(x)}{\sqrt{1-x^2}}\,dx+\frac{f(1)+f(-1)}{4}+O\left(\frac{\lVert \Delta f\rVert_2}{\N^c}\right)\\
\E|L_{\N}(f)|^2&=\frac{1}{2\pi}\int_{\mathbf{D}} \left|\nabla f_{\mathrm{sym}}\right|^2\, d^2z+\sum_{k\in \mathbf{Z}} |k|\left|\widehat{f_{\mathrm{sym}}|_{\partial\mathbf{D}}}(k)\right|^2\\
& \quad +\kappa_4 \left|\frac{1}{\pi}\int_{\mathbf{D}}f(z)d^2 z - \frac{1}{2\pi} \int_0^{2\pi} f(e^{i\theta}) d\theta\right|^2+O\left(\frac{\lVert \Delta f\rVert_2^2}{\N^c}\right),
\end{split}
\end{equation}
for some small fixed $c>0$, with
\[
f_{\mathrm{sym}}(z):=\frac{f(z)+f(\overline{z})}{2}.
\]
Note that the variance of linear eigenvalue statistics depends on the symmetrization of the test function with respect to the real axis; this reflects the fact that for matrices $X$ with real entries the spectrum is symmetric around the real axis.

We now consider $f=f_{\tau,\tau^*}$, with $f_{\tau,\tau^*}$ from (\ref{eq:defftau}). Note that for this choice of $f$ we have
\[
f_{\mathrm{sym}}(x+i y)= e^{itx}\cos(sy).
\]

We omit the computations of the expectation as they are completely analogous to \eqref{eq:first_int_expectation}--\eqref{eq:second_int_expectation}. Next, we compute the first term in the variance \eqref{eq:usefulrelreal}
\begin{equation}
    \begin{split}
    \label{eq:firststepreal}
      \frac{1}{2\pi}\int_{\mathbf{D}} \left|\nabla f_{\mathrm{sym}}\right|^2\, d^2z&= \frac{1}{2\pi}\int_{\mathbf{D}}\left[t^2(\cos (sy))^2+s^2 (\sin(sy))^2\right] \, dxdy \\
      &= \frac{1}{2\pi}\int_0^1\int_0^{2\pi}r\left[t^2(\cos (sr\sin \theta))^2+s^2 (\sin(sr\sin \theta))^2\right] \, d\theta dr.
    \end{split}
\end{equation}
Then, using that for any $a\in\mathbf{R}$ we have 
\[
\int_0^{2\pi} (\cos (a\sin x))^2\, dx= \pi (1+J_0(2a)), \qquad\quad \int_0^{2\pi} (\sin (a\sin x))^2\, dx= \pi (1-J_0(2a)),
\]
by \eqref{eq:firststepreal}, we conclude 
\begin{equation}
    \frac{1}{2\pi}\int_{\mathbf{D}} \left|\nabla f_{\mathrm{sym}}\right|^2\, d^2z=\frac{|\tau|^2}{4}+\frac{t^2-s^2}{8 s^2}\int_0^{2s} xJ_0(x)\, dx.
\end{equation}
Noticing that $\int_0^u xJ_0(x)\, dx= u J_1(u)$ from (i) of Fact \ref{fact:bessel_properties}, this concludes the computations of the first term in the variance in \eqref{eq:usefulrelreal}.

We now proceed with the second term in \eqref{eq:usefulrelreal}. Similarly to \eqref{eq:compint} and \eqref{eq:comph12}, we compute
\begin{equation}
    \begin{split}
\widehat{f_{\mathrm{sym}}|_{\partial\mathbf{D}}}(k)&=\frac{1}{2\pi} \int_0^{2\pi} e^{it \cos \theta}\cos (s\sin \theta) e^{-i\theta k}\, d\theta \\
&=\frac{e^{i\varphi_+ k}}{4\pi}\int_0^{2\pi} e^{i[-\theta k+\sqrt{t^2+s^2}\sin\theta]}\, d\theta+\frac{e^{i\varphi_- k}}{4\pi}\int_0^{2\pi} e^{i[-\theta k+\sqrt{t^2+s^2}\sin\theta]}\, d\theta \\
&=\frac{e^{i\varphi_+ k}+e^{i\varphi_- k}}{2}(-1)^{k+1}J_k(|\tau|),
    \end{split}
\end{equation}
where we defined $\varphi_{\pm}$ so that
\[
t \cos \theta \pm s \sin \theta = \sqrt{t^2 + s^2}\sin\round{\varphi_\pm + \theta},
\]
as done in (\ref{eq:double_angle}). We thus finally obtain
\[
\left|\widehat{f_{\mathrm{sym}}|_{\partial\mathbf{D}}}(k)\right|^2= (\sin(\varphi_+ k))^2 J_k(|\tau|)^2
\]
as a direct consequence of the fact that $\sin \varphi_+ = \sin \varphi_-$, whilst $\cos \varphi_+ = -\cos \varphi_-$. This concludes the proof of Theorem~\ref{thrm:dsff_leading_gen} in the real case as well. Finally, in order to conclude Corollary~\ref{thrm:dsff_leading}, we notice that also in the real case, squaring $e(\tau,\tau^*)$  in\eqref{eq:expvar} and using (iii) of Fact \ref{fact:bessel_properties}, we readily obtain $\left |\E \sum_{i} f_{\tau,\tau^*}(\sigma_i)\right |^2 \sim 4 \N^2J_1(|\tau|)^2/|\tau|^2$; the estimate of the variance is completely analogous to the complex case and so omitted.







\bibliographystyle{plain}
\bibliography{references}

\begin{thebibliography}{10}

\bibitem{abramowitz1988handbook}
Milton Abramowitz, Irene~A Stegun, and Robert~H Romer.
\newblock {Handbook of mathematical functions with formulas, graphs, and
  mathematical tables}, 1988.

\bibitem{borodin2009ginibre}
Alexei Borodin and Christopher~D Sinclair.
\newblock {The Ginibre ensemble of real random matrices and its scaling
  limits}.
\newblock {\em Communications in Mathematical Physics}, 291:177--224, 2009.

\bibitem{brezin1997spectral}
E~Br{\'e}zin and S~Hikami.
\newblock {Spectral form factor in a random matrix theory}.
\newblock {\em Physical Review E}, 55(4):4067, 1997.

\bibitem{byun2022progress}
Sung-Soo Byun and Peter~J Forrester.
\newblock {Progress on the study of the Ginibre ensembles I: GinUE}.
\newblock {\em arXiv preprint arXiv:2211.16223}, 2022.

\bibitem{byun2023progress}
Sung-Soo Byun and Peter~J Forrester.
\newblock {Progress on the study of the Ginibre ensembles II: GinOE and GinSE}.
\newblock {\em arXiv preprint arXiv:2301.05022}, 2023.

\bibitem{chou2011non}
Tom Chou, Kirone Mallick, and Royce~KP Zia.
\newblock {Non-equilibrium statistical mechanics: from a paradigmatic model to
  biological transport}.
\newblock {\em Reports on progress in physics}, 74(11):116601, 2011.

\bibitem{cipolloni2021fluctuation}
Giorgio Cipolloni, L{\'a}szl{\'o} Erd{\H{o}}s, and Dominik Schr{\"o}der.
\newblock {Fluctuation around the circular law for random matrices with real
  entries}.
\newblock {\em Electronic Journal of Probability}, 26:1--61, 2021.

\bibitem{cipolloni2022mesoscopic}
Giorgio Cipolloni, L{\'a}szl{\'o} Erd{\H{o}}s, and Dominik Schr{\"o}der.
\newblock {Mesoscopic Central Limit Theorem for non-Hermitian Random Matrices}.
\newblock {\em arXiv preprint arXiv:2210.12060}, 2022.

\bibitem{cipolloni2022optimal}
Giorgio Cipolloni, L{\'a}szl{\'o} Erd{\H{o}}s, and Dominik Schr{\"o}der.
\newblock {Optimal multi-resolvent local laws for Wigner matrices}.
\newblock {\em Electronic Journal of Probability}, 27:1--38, 2022.

\bibitem{cipolloni2022thermalisation}
Giorgio Cipolloni, L{\'a}szl{\'o} Erd{\H{o}}s, and Dominik Schr{\"o}der.
\newblock {Thermalisation for Wigner matrices}.
\newblock {\em Journal of Functional Analysis}, 282(8):109394, 2022.

\bibitem{cipolloni2019central}
Giorgio Cipolloni, L{\'a}szl{\'o} Erd{\H{o}}s, and Dominik Schr{\"o}der.
\newblock Central limit theorem for linear eigenvalue statistics of
  non-hermitian random matrices.
\newblock {\em Communications on Pure and Applied Mathematics},
  76(5):946--1034, 2023.

\bibitem{cipolloni2023spectral}
Giorgio Cipolloni, L{\'a}szl{\'o} Erd{\H{o}}s, and Dominik Schr{\"o}der.
\newblock {On the spectral form factor for random matrices}.
\newblock {\em Communications in Mathematical Physics}, pages 1--36, 2023.

\bibitem{cipolloni2023entanglement}
Giorgio Cipolloni and Jonah Kudler-Flam.
\newblock {Entanglement Entropy of Non-Hermitian Eigenstates and the Ginibre
  Ensemble}.
\newblock {\em Physical Review Letters}, 130(1):010401, 2023.

\bibitem{cipolloni2023non}
Giorgio Cipolloni and Jonah Kudler-Flam.
\newblock {Non-Hermitian Hamiltonians Violate the Eigenstate Thermalization
  Hypothesis}.
\newblock {\em arXiv preprint arXiv:2303.03448}, 2023.

\bibitem{cotler2017black}
Jordan~S Cotler, Guy Gur-Ari, Masanori Hanada, Joseph Polchinski, Phil Saad,
  Stephen~H Shenker, Douglas Stanford, Alexandre Streicher, and Masaki Tezuka.
\newblock {Black holes and random matrices}.
\newblock {\em Journal of High Energy Physics}, 2017(5):1--54, 2017.

\bibitem{cummer2016controlling}
Steven~A Cummer, Johan Christensen, and Andrea Al{\`u}.
\newblock {Controlling sound with acoustic metamaterials}.
\newblock {\em Nature Reviews Materials}, 1(3):1--13, 2016.

\bibitem{deng2010exciton}
Hui Deng, Hartmut Haug, and Yoshihisa Yamamoto.
\newblock {Exciton-polariton bose-einstein condensation}.
\newblock {\em Reviews of modern physics}, 82(2):1489, 2010.

\bibitem{el2018non}
Ramy El-Ganainy, Konstantinos~G Makris, Mercedeh Khajavikhan, Ziad~H
  Musslimani, Stefan Rotter, and Demetrios~N Christodoulides.
\newblock {Non-Hermitian physics and PT symmetry}.
\newblock {\em Nature Physics}, 14(1):11--19, 2018.

\bibitem{feng2017non}
Liang Feng, Ramy El-Ganainy, and Li~Ge.
\newblock {Non-Hermitian photonics based on parity--time symmetry}.
\newblock {\em Nature Photonics}, 11(12):752--762, 2017.

\bibitem{forrester2021differential}
Peter~J Forrester.
\newblock {Differential identities for the structure function of some random
  matrix ensembles}.
\newblock {\em Journal of Statistical Physics}, 183(2):33, 2021.

\bibitem{forrester2021quantifying}
Peter~J Forrester.
\newblock {Quantifying Dip--Ramp--Plateau for the Laguerre Unitary Ensemble
  Structure Function}.
\newblock {\em Communications in Mathematical Physics}, 387(1):215--235, 2021.

\bibitem{forrester2022dip}
Peter~J Forrester, Mario Kieburg, Shi-Hao Li, and Jiyuan Zhang.
\newblock {Dip-ramp-plateau for Dyson Brownian motion from the identity on $ U
  (N) $}.
\newblock {\em arXiv preprint arXiv:2206.14950}, 2022.

\bibitem{forrester2007eigenvalue}
Peter~J Forrester and Taro Nagao.
\newblock {Eigenvalue statistics of the real Ginibre ensemble}.
\newblock {\em Physical review letters}, 99(5):050603, 2007.

\bibitem{forrester2009matrix}
Peter~J Forrester and Eric~M Rains.
\newblock {Matrix averages relating to Ginibre ensembles}.
\newblock {\em Journal of Physics A: Mathematical and Theoretical},
  42(38):385205, 2009.

\bibitem{forrester1999fluctuation}
PJ~Forrester.
\newblock {Fluctuation formula for complex random matrices}.
\newblock {\em Journal of Physics A: Mathematical and General}, 32(13):L159,
  1999.

\bibitem{garcia2022dominance}
Antonio~M Garc{\'\i}a-Garc{\'\i}a, Yiyang Jia, Dario Rosa, Jacobus~JM
  Verbaarschot, et~al.
\newblock {Dominance of Replica Off-Diagonal Configurations and Phase
  Transitions in a P T Symmetric Sachdev-Ye-Kitaev Model}.
\newblock {\em Physical Review Letters}, 128(8):081601, 2022.

\bibitem{garcia2018universality}
Antonio~M Garc{\'\i}a-Garc{\'\i}a, Yiyang Jia, Jacobus~JM Verbaarschot, et~al.
\newblock {Universality and Thouless energy in the supersymmetric
  Sachdev-Ye-Kitaev model}.
\newblock {\em Physical Review D}, 97(10):106003, 2018.

\bibitem{garcia2023universality}
Antonio~M Garc{\'\i}a-Garc{\'\i}a, Lucas S{\'a}, and Jacobus~JM Verbaarschot.
\newblock {Universality and its limits in non-Hermitian many-body quantum chaos
  using the Sachdev-Ye-Kitaev model}.
\newblock {\em Physical Review D}, 107(6):066007, 2023.

\bibitem{garcia2017analytical}
Antonio~M Garc{\'\i}a-Garc{\'\i}a and Jacobus~JM Verbaarschot.
\newblock {Analytical spectral density of the Sachdev-Ye-Kitaev model at finite
  N}.
\newblock {\em Physical Review D}, 96(6):066012, 2017.

\bibitem{ghosh2022spectral}
Soumi Ghosh, Sparsh Gupta, and Manas Kulkarni.
\newblock {Spectral properties of disordered interacting non-Hermitian
  systems}.
\newblock {\em Physical Review B}, 106(13):134202, 2022.

\bibitem{ginibre1965statistical}
Jean Ginibre.
\newblock {Statistical ensembles of complex, quaternion, and real matrices}.
\newblock {\em Journal of Mathematical Physics}, 6(3):440--449, 1965.

\bibitem{grobe1989universality}
Rainer Grobe and Fritz Haake.
\newblock {Universality of cubic-level repulsion for dissipative quantum
  chaos}.
\newblock {\em Physical review letters}, 62(25):2893, 1989.

\bibitem{grobe1988quantum}
Rainer Grobe, Fritz Haake, and Hans-J{\"u}rgen Sommers.
\newblock Quantum distinction of regular and chaotic dissipative motion.
\newblock {\em Physical review letters}, 61(17):1899, 1988.

\bibitem{hamazaki2020universality}
Ryusuke Hamazaki, Kohei Kawabata, Naoto Kura, and Masahito Ueda.
\newblock {Universality classes of non-Hermitian random matrices}.
\newblock {\em Physical Review Research}, 2(2):023286, 2020.

\bibitem{jia2020spectral}
Yiyang Jia and Jacobus~JM Verbaarschot.
\newblock {Spectral fluctuations in the Sachdev-Ye-Kitaev model}.
\newblock {\em Journal of High Energy Physics}, 2020(7):1--59, 2020.

\bibitem{kawabata2023entanglement}
Kohei Kawabata, Tokiro Numasawa, and Shinsei Ryu.
\newblock {Entanglement phase transition induced by the non-hermitian skin
  effect}.
\newblock {\em Physical Review X}, 13(2):021007, 2023.

\bibitem{kopel2015linear}
Phil Kopel.
\newblock {Linear statistics of non-Hermitian matrices matching the real or
  complex Ginibre ensemble to four moments}.
\newblock {\em arXiv preprint arXiv:1510.02987}, 2015.

\bibitem{leviandier1986fourier}
Luc Leviandier, Maurice Lombardi, R{\'e}mi Jost, and Jean~Paul Pique.
\newblock {Fourier transform: A tool to measure statistical level properties in
  very complex spectra}.
\newblock {\em Physical review letters}, 56(23):2449, 1986.

\bibitem{li2021spectral}
Jiachen Li, Toma{\v{z}} Prosen, and Amos Chan.
\newblock {Spectral statistics of non-hermitian matrices and dissipative
  quantum chaos}.
\newblock {\em Physical review letters}, 127(17):170602, 2021.

\bibitem{ma2016acoustic}
Guancong Ma and Ping Sheng.
\newblock {Acoustic metamaterials: From local resonances to broad horizons}.
\newblock {\em Science advances}, 2(2):e1501595, 2016.

\bibitem{marchetti2013hydrodynamics}
M~Cristina Marchetti, Jean-Fran{\c{c}}ois Joanny, Sriram Ramaswamy,
  Tanniemola~B Liverpool, Jacques Prost, Madan Rao, and R~Aditi Simha.
\newblock {Hydrodynamics of soft active matter}.
\newblock {\em Reviews of modern physics}, 85(3):1143, 2013.

\bibitem{matsoukas2023quantum}
Apollonas~S Matsoukas-Roubeas, Toma{\v{z}} Prosen, and Adolfo del Campo.
\newblock {Quantum Chaos and Coherence: Random Parametric Quantum Channels}.
\newblock {\em arXiv preprint arXiv:2305.19326}, 2023.

\bibitem{matsoukas2023non}
Apollonas~S Matsoukas-Roubeas, Federico Roccati, Julien Cornelius, Zhenyu Xu,
  Aurelia Chenu, and Adolfo del Campo.
\newblock {Non-Hermitian Hamiltonian deformations in quantum mechanics}.
\newblock {\em Journal of High Energy Physics}, 2023(1):1--31, 2023.

\bibitem{may1972will}
Robert~M May.
\newblock {Will a large complex system be stable?}
\newblock {\em Nature}, 238:413--414, 1972.

\bibitem{mehta2004random}
Madan~Lal Mehta.
\newblock {\em {Random matrices}}.
\newblock Elsevier, 2004.

\bibitem{muller2012engineered}
Markus M{\"u}ller, Sebastian Diehl, Guido Pupillo, and Peter Zoller.
\newblock {Engineered open systems and quantum simulations with atoms and
  ions}.
\newblock In {\em Advances in Atomic, Molecular, and Optical Physics},
  volume~61, pages 1--80. Elsevier, 2012.

\bibitem{nourdin2010universal}
Ivan Nourdin and Giovanni Peccati.
\newblock {Universal Gaussian fluctuations of non-Hermitian matrix ensembles:
  from weak convergence to almost sure CLTs}.
\newblock {\em arXiv preprint arXiv:1002.1212}, 2010.

\bibitem{okuyama2019spectral}
Kazumi Okuyama.
\newblock {Spectral form factor and semi-circle law in the time direction}.
\newblock {\em Journal of High Energy Physics}, 2019(2):1--16, 2019.

\bibitem{o2016central}
Sean O’Rourke and David Renfrew.
\newblock {Central limit theorem for linear eigenvalue statistics of elliptic
  random matrices}.
\newblock {\em Journal of Theoretical Probability}, 29:1121--1191, 2016.

\bibitem{rider2006gaussian}
Brian Rider and Jack~W Silverstein.
\newblock {Gaussian fluctuations for non-Hermitian random matrix ensembles}.
\newblock {\em Annals of Probability}, 2006.

\bibitem{rider2007noise}
Brian Rider and B{\'a}lint Vir{\'a}g.
\newblock {The noise in the circular law and the Gaussian free field}.
\newblock {\em International Mathematics Research Notices}, 2007, 2007.

\bibitem{ritsch2013cold}
Helmut Ritsch, Peter Domokos, Ferdinand Brennecke, and Tilman Esslinger.
\newblock {Cold atoms in cavity-generated dynamical optical potentials}.
\newblock {\em Reviews of Modern Physics}, 85(2):553, 2013.

\bibitem{saad2018semiclassical}
Phil Saad, Stephen~H Shenker, and Douglas Stanford.
\newblock {A semiclassical ramp in SYK and in gravity}.
\newblock {\em arXiv preprint arXiv:1806.06840}, 2018.

\bibitem{sieberer2016keldysh}
Lukas~M Sieberer, Michael Buchhold, and Sebastian Diehl.
\newblock {Keldysh field theory for driven open quantum systems}.
\newblock {\em Reports on Progress in Physics}, 79(9):096001, 2016.

\bibitem{song2019non}
Fei Song, Shunyu Yao, and Zhong Wang.
\newblock {Non-Hermitian skin effect and chiral damping in open quantum
  systems}.
\newblock {\em Physical review letters}, 123(17):170401, 2019.

\bibitem{tao2015random}
Terence Tao and Van Vu.
\newblock {Random matrices: universality of local spectral statistics of
  non-Hermitian matrices}.
\newblock {\em Annals of Probability}, 2015.

\end{thebibliography}

\end{document}